\documentclass[twocolumn,showpacs,amsmath,amssymb,prl]{revtex4}
\usepackage[all,2cell]{xy}
\usepackage{color}
\usepackage{nicefrac}
\usepackage[latin1]{inputenc}
\usepackage{epsfig}
\usepackage{circuitikz}
\usepackage{amsthm}
\usepackage{graphicx}
\usepackage{amsmath}
\newcommand{\niton}{\not\mathrel{\text{\reflectbox{$\in$}}}}

\newcommand*{\Scale}[2][4]{\scalebox{#1}{$#2$}}
\newcommand{\tc}[1]{{\color{black} #1}}
\newcommand{\bs}[1] {\boldsymbol{#1}}
\newcommand{\ba}{\begin{array}}
\newcommand{\ea}{\end{array}}

\newcommand{\C}{\mathcal{C}}
\newcommand{\Q}{\widetilde{F}}
\newcommand{\F}{F}

\newcommand{\W}{\mathbb{W}}
\newcommand{\M}{\mathbb{M}}
\newcommand{\U}{\mathbb{U}}
\newcommand{\R}{\mathbb{R}}
\renewcommand{\P}{\mathbb{P}}
\newcommand{\prob}{P}
\newcommand{\st}{{\mathrm{st}}}
\renewcommand{\phi}{J}

\newtheorem{theorem}{Theorem}

\begin{document}

\title{Effective Thermodynamics for a Marginal Observer}

\author{Matteo Polettini}
\author{Massimiliano Esposito}
\email{matteo.polettini@uni.lu}
\affiliation{Physics and Materials Science Research Unit, University of Luxembourg, Campus
Limpertsberg, 162a avenue de la Fa\"iencerie, L-1511 Luxembourg (Luxembourg)} 
 
\date{\today}

\begin{abstract}
Thermodynamics is usually  formulated on the presumption that the observer has complete information about the system he/she deals with: no parasitic current, exact evaluation of the forces that drive the system. For example, the acclaimed Fluctuation Relation (FR), relating the probability of time-forward  and time-reversed trajectories, assumes that the measurable transitions suffice to characterize the process as Markovian (in our case, a continuous-time jump process). However, most often the observer only measures a {\it marginal} current. We show that he/she will nonetheless produce an {\it effective} description that does not dispense with the fundamentals of thermodynamics, including the FR and the 2nd law. Our results stand on the mathematical construction of a hidden time reversal of the dynamics, and on the physical requirement that the observed current only accounts for a single transition in the configuration space of the system. We employ a simple abstract example to illustrate our results and to discuss the feasibility of generalizations.
\end{abstract} 


\maketitle

Nonequilibrium thermodynamics is a discourse about a set of currents  \tc{$\bs{\phi} = (\phi_\alpha)_{\alpha =1}^n$} flowing across an open  system and about their cost, viz. the forces or affinities  \tc{$\bs{F} = (\F_\alpha)_{\alpha =1}^n$} that need to be exerted to sustain such flows. Statistical Mechanics makes currents into random variables  whose joint probability \tc{$P(\bs{\phi})$} is determined by the occurrence of an underlying Markovian trajectory $\omega$ in configuration space. Be \tc{they} electric currents and voltages, heat flows and temperature gradients, etc., thermodynamics establishes certain universal truths, among which, in a progression that covers the 19th to 21st century time frame: ({\it 2nd law}) the average steady-state rate at which entropy is delivered to the environment  \tc{$\bs{F} \cdot \langle\bs{\phi} \rangle$} is non-negative;
({\it FDR}) near equilibrium a perturbation of the currents leads to \emph{D}issipation, which is \emph{R}elated to the current's \emph{F}luctuations at equilibrium; ({\it FR}) the probability of observing positive \emph{F}luctuations of the entropy \tc{production} is exponentially favoured if \emph{R}elated to negative ones \cite{andrieux}:
\begin{align}
\tc{\frac{P(\bs{\phi})}{P(-\bs{\phi})} = \exp \left(t \, \bs{F} \cdot  \bs{\phi} \right) . \label{eq:fr}}
\end{align}
Notoriously, the FR holds in the long-time limit, but it can be formulated at all times upon a proper choice of the equilibrium distribution from which the initial configuration is sampled \cite{udotrans,bulnestrans,polettiniFT}.

The above results require that all possible sources of dissipation are known. However, most often such a complete description is impossible to achieve, neither experimentally (\tc{leakages}) nor theoretically (coarse-grained degrees of freedom \cite{espocg,bo}). How can then thermodynamics be established in a consistent way? In this letter we consider a {\it marginal} observer who only monitors one current $\phi = \phi_{\alpha\equiv1}$ flowing between two discrete configurations, and who controls a parameter that only affects the rates at which transitions between \tc{those} configurations occur. We show that \tc{the} observer would guess an  {\it effective} affinity such that the following occurs: (i) a {\it marginal} FR for the observed current holds at large times, for any given statistical distribution of the initial state, or at all times upon a proper choice of the initial \tc{distribution} (see also Refs.\,\cite{hartich,sagawa,rosinberg,crooks} for alternative formulations of a marginal FR); (ii) The effective affinity has a clear interpretation as the counteracting force needed to {\it stall} the current, i.e. to make it vanish;  (iii) there exists a simple operational procedure to determine the effective affinity; (iv) the FDR at stalling derived in Ref.\,\cite{altaner} follows as a consequence; (v) the marginal  FR requires a new notion of ``hidden time-reversal'' (HTR), an interesting mathematical construct rich in properties, though lacking a clear operational implementation. In some sense HTR inverts the dynamics in the hidden subspace (see Ref.\,\cite{qian} for an analog for underdamped Langevin dynamics).

Our conclusions stand on several mathematical propositions. We only include in the letter the outline of those proofs whose elements are necessary to interpret the results, leaving details to the Supplemental Material.

\paragraph{Story of a marginal observer.} Consider a {\it gedankenobserver} who can operate a measurement apparatus that only resolves two configurations $i = 1,2$. To him, the rest of the system is a black box:
\begin{align}
\begin{array}{c}\xymatrix{1 \ar@{-}[rr] \ar@{-}[dr] \ar@{..}@/_/[dr] \ar@{-}@/^/[dr] & & 2  \ar@{..}@/^/[dl]  \ar@{-}@/_/[dl]  \ar@{-}[dl] \\
& \rule{1cm}{1cm}  & 
}\end{array}.
\end{align}
The observer controls the rate $w_{12}(x)$ at which, given that the apparatus recorded $2$, the next measurement is $1$, and similarly for $w_{21}(x)$. The observer also records how often configuration $2$ occurs with respect to $1$. This latter piece of information is quantified by the ratio $p_1(x)/p_2(x)$, where $\vec{p}\, (x)$ is the probability of being anywhere in configuration space, including in the black box, whose population $p_\blacksquare(x)$ is unknown. Finally, the observer handles a \tc{parameter} $x$ by which he tunes \tc{the values of the rates}. We assume that rates inside the black box do not depend on $x$, and that $x$ satisfies
\begin{align}
\frac{w_{12}(x)}{w_{21}(x)} = \frac{p^{\mathrm{eq}}_1(x)}{p^{\mathrm{eq}}_2(x)} = e^x
\end{align}
where $\vec{p}^{\;\mathrm{eq}}$ \tc{is} the detailed balanced steady state \tc{when the connections to the black box are all ``cut'' and and the edge is disconnected from the network, and as a consequence} no net current \tc{is} measured along $1-2$.
\tc{Let us give a more physical insight into the parametrization via two examples. Rates of systems in contact with heat reservoirs satisfy {\it local detailed balance}, e.g. $\log(w_{12}/w_{21}) =  \nicefrac{(\varepsilon_2 - \varepsilon_1)}{k_B T_{12}}$, with $\varepsilon_1,\varepsilon_2$ energy levels and $T_{12}$ the temperature of the environmental degrees of freedom that interact with the \tc{observable transition, while other edges are at different temperatures}. The parameter in this case would affect the {\it external} temperature of the reservoir affecting transition $1-2$ and not the {\it internal} energy levels, which do affect other transitions. A similar example is that of a specific chemical reaction $\mathrm{X}_1 + \mathrm{C} \leftrightharpoons \mathrm{X}_2$ belonging to a larger network of chemical reactions, where an internal reactant $\mathrm{X}_1$ is coupled to a {\it chemostat} $\mathrm{C}$ at fixed but controllable concentration \cite{poleCN}. In this case, it is only possible to modify the concentration of $\mathrm{C}$, while the internal ``energy levels''  are unmodifiable integer particle numbers} \footnote{\tc{However, as regards our theory, the case of chemical reaction networks is significantly more complicated because currents are defined along several edges, and because the topology of the chemical network affects its thermodynamics \cite{polewachtel}.}}.

If detailed balance is not observed and the observer measures a different value of $p_1(x)/p_2(x) \neq p_1^{\mathrm{eq}}(x)/p_2^{\mathrm{eq}}(x)$, a current $\langle\phi\rangle(x) = w_{12}(x)p_2(x) - w_{21}(x)p_1(x)$ flows along $1-2$. The observer then needs to formulate a minimal model that is compatible with this observation. The simplest possible setup is
\begin{align}
\ba{c}\xymatrix{1 \ar@{-}@/^1pc/[r]^{w(x)}   \ar@{-}@/_1pc/[r]_{\widetilde{w}}  & 2} \ea.
\end{align}
In this minimal model the black box is responsible of returning an event at $1$ or $2$ at some effective rates $\widetilde{w}_{12},\widetilde{w}_{21}$. While this is not a viable approximation of the black-box dynamics, which is described by more advanced projection techniques \cite{sollich}, it is enough to replicate the average  steady-state measurements of our observer.

Our first result is the determination of the effective rates. Notice that the minimal model must satisfy the global detailed balance condition obtained by lumping the transitions  $w$ and $\widetilde{w}$:
\begin{align}
\frac{w_{21}(x) + \widetilde{w}_{21}}{w_{12}(x) + \widetilde{w}_{12}} = \frac{p_2(x)}{p_1(x)}. \label{eq:condition}
\end{align}
We will now compare this quantity with the truth-of-matter of the complete system. To illustrate our results, we employ the following example:
\begin{align}
\ba{c}\xymatrix{
1 \ar@{-}[r] & 2  \ar@{-}[d]  \\
 \ar@{-}[u]  4 \ar@{-}[r] & 3  \ar@{-}[ul] 
} \ea, \qquad
\ba{c}\xymatrix{
1 \ar@_{->}@<-0.4mm>[r] \ar@^{-}@<0.0mm>[r] \ar@^{->}@<0.4mm>[r] & 2 
\ar@_{->}@<-0.4mm>[d] \ar@^{-}@<0.0mm>[d] \ar@^{->}@<0.4mm>[d]  \\
  \ar@_{->}@<-0.2mm>[u]   \ar@^{->}@<0.2mm>[u]    4  & 3  \ar@{->}[ul]     \ar@_{->}@<-0.2mm>[l]   \ar@^{->}@<0.2mm>[l] 
}\ea \stackrel{\mathrm{TR}\;}{\longrightarrow}  \ba{c}\xymatrix{
1 \ar@_{<-}@<-0.4mm>[r] \ar@^{-}@<0.0mm>[r] \ar@^{<-}@<0.4mm>[r] & 2 
\ar@_{<-}@<-0.4mm>[d] \ar@^{-}@<0.0mm>[d] \ar@^{<-}@<0.4mm>[d]  \\
  \ar@_{<-}@<-0.2mm>[u]   \ar@^{<-}@<0.2mm>[u]    4  & 3  \ar@{<-}[ul]     \ar@_{<-}@<-0.2mm>[l]   \ar@^{<-}@<0.2mm>[l] 
}\ea. \label{eq:model}
\end{align}
The first diagram depicts the topology, where the graph's {\it edges} connect {\it sites} (configurations). The second \tc{is a pictorial illustration of a possible} steady state \tc{in configuration space, each line representing a given amount of current}, obeying Kirchhoff's law of current conservation at the sites of the graph. The third is its time reversed (for later reference).

We assume that the full system evolves by the Master Equation $d\vec{p}(t)/dt = \W \vec{p}(t)$, with generator $\W_{ij} := w_{ij} - w_i \delta_{i,j} $, where rates $w_{ji},w_{ij} > 0$ are positive along all edges of the graph, and $w_i := \sum_{k\neq i} w_{ki}$ are the exit rates \cite{gaveau}. The steady state $\vec{p}$ is the unique null vector of $\W$. Let $\W_{(j_1,\ldots,j_n|i_1,\ldots,i_m)}$ be the matrix obtained by removing rows $j_1,\ldots,j_n$ and columns $i_1,\ldots,i_m$. Then, {\bf Theorem 1} states that the effective rates are given by
\begin{align}
\widetilde{w}_{12} & \tc{ = \frac{\det \W^{\mathrm{st}}_{(2\tc{|}1)}}{\det \W_{(1,2|1,2)}} = \frac{ \hspace{-.4cm}
\Scale[0.5]{
   \ba{c}\xymatrix{ & \ar[d] \\ \ar[u] & \ar[l]} \ea 
+ \ba{c}\xymatrix{ & \ar@{->}[d] \\ \ar[u] & \ar@{->}[ul] } \ea
+ \ba{c}\xymatrix{   & \ar[d] \\  \ar[r] & \ar[ul] } \ea} }
{ \Scale[0.5]{
   \ba{c}\xymatrix{ &  \ar@{<-}[d] \\  \ar[u] & } \ea
+ \ba{c}\xymatrix{ & \ar@{<-}[d] \\ & \ar@{<-}[l]} \ea
+ \ba{c}\xymatrix{ &  \\ \ar[u] & \ar[l]} \ea
+ \ba{c}\xymatrix{ & \\  \ar[r] & \ar[ul] } \ea
+ \ba{c}\xymatrix{ &    \\  \ar[u] & \ar[ul]  } \ea}} }
\end{align}
and similarly for $\widetilde{w}_{21}$, where $\W^{\mathrm{st}}$ is the {\it stalling} generator obtained from $\W$ by setting $w_{12} \equiv w_{21} \equiv 0$. \tc{Each illustrative diagram on the right-hand side represents a term of the determinant, whereby the multiplication of rates along oriented arrows is implied, e.g. $ \Scale[0.4]{\ba{c}\xymatrix{ & \ar[d] \\ \ar[u] & \ar[l]} \ea  } = w_{14} w_{43} w_{32}$}. Notice that, consistently with the ``no-$x$-in-$\blacksquare$'' assumption, the effective rates do not depend on $x$. The proof of {\bf Theorem 1} relies on some facts in algebraic graph theory, in particular, that the steady-state probability can be written in terms of the minors of the generator \cite{gaveau}, that by the all-minors matrix-tree theorem \cite{chaiken} minors can be given a graphical interpretation in terms of spanning trees, and that by the deletion/contraction principle \cite{sokal} one can single out the trees of the stalling generator and the trees that pass by edge $1-2$.

\paragraph{Effective affinity and stalling.}

Supporting a steady-state current has a thermodynamic cost. According to the theory of Hill-Schnakenberg \cite{hill,schnak}, which is the Markov-process analog of Kirchhoff's mesh analysis for electrical circuits, this cost is quantified by cycle affinities, i.e. the log-ratio of products of rates along an independent set of oriented cycles, in the two directions. In the complete system, there are a number of ``real'' affinities
\begin{align}
\F_\alpha(x) =  \log \prod_{i \gets j \in \mathcal{C}_\alpha} \frac{w_{ij}}{w_{ji}}.
\end{align} Notice that these affinities are either in the form $\F_\alpha(x) = x - x^{\mathrm{ref}}_\alpha$, with $ x^{\mathrm{ref}}_\alpha$ some reference value \cite{polettiniFT}, if cycle $\mathcal{C}_\alpha$ includes the observable transition, or else do not depend on $x$. Hence a variation of $x$ corresponds to a variation of the relevant thermodynamic forces.

In the case of the minimal model, there is only one cycle $1 \longleftarrow 2 \stackrel{\;\sim}{\longleftarrow} 1$ with {\it effective affinity} 
\begin{align}
\Q(x) := \log \frac{w_{12}(x)\widetilde{w}_{21}}{w_{21}(x)\widetilde{w}_{12}}  =  x - x^{\mathrm{st}}
\end{align}
where we define the stalling value of the parameter as
\begin{align}
x^{\mathrm{st}} =  \log \frac{ 
\Scale[0.6]{
 \ba{c}\xymatrix{ & \ar@{<-}[d] \\ \ar@{<-}[u] & \ar@{<-}[l]} \ea
+ \ba{c}\xymatrix{ & \ar@{<-}[d] \\ \ar[u] & \ar@{<-}[ul] } \ea
+ \ba{c}\xymatrix{   & \ar@{<-}[d] \\  \ar[r] & \ar@{<-}[ul]  } \ea
 } }{
\Scale[0.6]{
  \ba{c}\xymatrix{ & \ar[d] \\ \ar[u] & \ar[l]} \ea 
+ \ba{c}\xymatrix{ & \ar@{->}[d] \\ \ar[u] & \ar@{->}[ul] } \ea
+ \ba{c}\xymatrix{   & \ar[d] \\  \ar[r] & \ar[ul] } \ea } } = \log \frac{p^{\mathrm{st}}_1}{p^{\mathrm{st}}_2}.
\end{align}
In this latter expression we recognized $\vec{p}^{\;\mathrm{st}}$ as the steady state of the stalling system where edge $1-2$ is removed altogether,  $\W^{\mathrm{st}} \vec{p}^{\;\mathrm{st}} = 0$. The effective affinity then reads
\begin{align}
\Q = \log \frac{w_{12} \,p_2^{\mathrm{st}}}{w_{21} \, p_1^{\mathrm{st}}}
\end{align}
and it can then be interpreted as the force that is exerted if the system is prepared in the stalling state and transition $1-2$ is suddenly turned on or, conversely, as the extra force that should be applied to stall the current, \tc{as already observed by Qian \cite{qian2}, who dubbed it ``isometric force''}. Such a stalling state is reached when $x = x^{\mathrm{st}}$, where both the effective  affinity and the marginal current vanish, $\langle \phi\rangle(x^{\mathrm{st}}) = 0$. As an important consequence, the effective affinity can be determined operationally by a simple calibration procedure, and not only based on an abstract mathematical construction. The procedure is as follows: tune $x$ to the value $x^\st$ that makes the current vanish along the observable edge. Then $\widetilde{F} = x-x^{\st}$. Despite \tc{the fact that according} to the local observer stalling looks like an equilibrium state, yet in the black box currents of arbitrary magnitude might be flowing \cite{lacoste}.
 
\paragraph{Fluctuations.} So far we considered a minimal model reproducing the observer's steady-state average measurements and predicting the effective thermodynamic cost of sustaining the observable current. The question we now address is whether this marginal description is robust as it comes to fluctuations of the current. 

Establishing the complete FR Eq.\,(\ref{eq:fr}) requires a notion of time-reversed dynamics $\overline{\W} := \P \, \W^T \P^{-1}$, where ${}^T$ denotes matrix transposition and $\P := \mathrm{diag}\,\{p_i\}_i$ is the diagonal matrix whose entries are the steady-state probabilities. Letting $\overline{P}$ be the probability measure associated to the time-reversed dynamics, under time reversal currents change sign in probability $\overline{P}(\tc{\bs{\phi}}) = P(\tc{-\bs{\phi}})$ (see the right-hand side of  Eq.\,(\ref{eq:model}) for an illustration).

We now focus on the marginal probability of the observable current $P(\phi) = \int P(\tc{\bs{\phi}}) \prod_{\alpha \neq 1} d\phi_\alpha $, which, in general, does not satisfy a FR. Yet, we are able to construct a {\it hidden time-reversal dynamics} (HTR) $\widetilde{\W}$ such that the following {\it marginal FR} holds ({\bf Theorem 3})
\begin{align}
\frac{P(\phi)}{\widetilde{P}(-\phi)} = \exp\left( t\,\Q \phi\right), \label{eq:MFR}
\end{align}
where, quite remarkably, the effective affinity is precisely  that estimated by the local observer.

Before introducing the HTR dynamics, let us draw some consequences. Importantly, forward and HTR dynamics differ, $P(\phi) \neq \widetilde{P}(\phi)$, which leads to a substantial difference with respect to  Eq.\,(\ref{eq:fr}). However, it might be that for a  ``fast box'', where there is a time-scale separation between observable and hidden degrees of freedom, $P(\phi) \approx \widetilde{P}(\phi)$, as the case  analyzed in Ref.\,\cite{bulnes} suggests.
While we do derive explicit expressions for the rates of the HTR dynamics, in practice microengineering such rates in the absence of a clear operational procedure seems to be unfeasible, so that the marginal FR might be out of experimental reach. However, \tc{standard manipulations of Eq.\.(\ref{eq:MFR})} lead to the {\it marginal integral FR}
\begin{align}
\Big\langle \exp \left[- t\, \Q(x) \phi\right] \Big\rangle(x) = 1, \label{eq:IFR}
\end{align}
which only depends on the original dynamics, while the HTR  probability $\widetilde{P}$ is traced out. Hence, the latter relation is experimentally accessible. Furthermore, using Jensen's inequality we  obtain the effective second law $ \Q(x) \langle \phi \rangle(x)  \geq 0$; \tc{bounds on the marginal dissipation have been derived in Ref.\,\cite{bisker}}. Taking the second derivative with respect to $x$ and evaluating at stalling, we obtain the nonequilibrium FDR \tc{already} proven in Ref.\,\cite{altaner}:
\begin{align}
\left. \frac{\partial \langle \phi \rangle}{\partial \Q}\right|_{\mathrm{st}}  =  \left.\frac{\tc{\mathrm{var}^t}({\phi})}{2}\right|_{\mathrm{st}}.
\end{align}
Here, $\mathrm{var}^t$ is a properly time-scaled variance, considering that all cumulants of currentlike observables grow linearly in time (see Ref.\,\cite{altaner15a} for details). We also used the fact that, by virtue of the parametrization, the explicit derivative $\nicefrac{d}{dx}$ is equivalent to the implicit variation of the effective affinity $\partial_{\Q}$. Let us emphasize that this is not just a technical subtlety: it underlies the possibility of interpreting our results in operational terms.

\paragraph{Hidden time reversal.}

We now specify the HTR dynamics, by defining the operator
\begin{multline}
\widetilde{\W}  :=  \W  - \W_{\mathrm{st}} +\overline{\W}_{\mathrm{st}}\label{eq:hidden} \\ = \left(\ba{cccc} - w_1 & w_{12}  & w_{31}\frac{p_1^\st}{p_3^\st} & w_{41}\frac{p_1^\st}{p_4^\st}  \\  w_{21} & -w_2 & w_{32}\frac{p_2^\st}{p_3^\st}  & 0 \\  w_{13}\frac{p_3^\st}{p_1^\st}  & w_{23} \frac{p_3^\st}{p_2^\st}  & - w_3 & w_{43}\frac{p_3^\st}{p_4^\st}  \\  w_{14} \frac{p_4^\st}{p_1^\st} & 0 & w_{34} \frac{p_4^\st}{p_3^\st}  & -w_4 \ea \right),
\end{multline}
where $\overline{\W}_{\mathrm{st}} =  \P_\mathrm{st} \W_{\mathrm{st}}^T \P_\mathrm{st}^{-1} $, and $\P_\mathrm{st}$ is the diagonal matrix with entries given by the components of the stalling steady state.  Crucially, $\widetilde{\W}$ can be proven to be a Markov jump-process generator ({\bf Theorem 2}). \tc{Exit rates (diagonal elements) are the same as in the original dynamics $\W$}. Furthermore, the upper-left 2 $\times$ 2 block (the only nonvanishing block of $\W - \W_{\mathrm{st}})$ is unchanged while the rest of the generator undergoes time reversal, \tc{with respect to its proper stalling steady state}. So, in a way, the HTR is the best attempt to invert the dynamics everywhere but in the observable subsystem, whence its name. \tc{The following diagram is a pictorial illustration of HTR of the steady currents:}
\begin{align}
\ba{c}\xymatrix{
1 \ar@_{->}@<-0.4mm>[r] \ar@^{-}@<0.0mm>[r] \ar@^{->}@<0.4mm>[r] & 2 
\ar@_{->}@<-0.4mm>[d] \ar@^{-}@<0.0mm>[d] \ar@^{->}@<0.4mm>[d]  \\
  \ar@_{->}@<-0.2mm>[u]   \ar@^{->}@<0.2mm>[u]    4  & 3  \ar@{->}[ul]     \ar@_{->}@<-0.2mm>[l]   \ar@^{->}@<0.2mm>[l] 
}\ea
\stackrel{\mathrm{HTR}}{\longrightarrow}
\ba{c}\xymatrix{
1 \ar@_{->}@<-0.2mm>[r] \ar@^{->}@<0.2mm>[r] & 2 
\ar@_{->}@<-0.2mm>[d] \ar@^{->}@<0.2mm>[d]  \\
  \ar@{<-}[u]  4  & 3    \ar@{<-}[l]  \ar@_{->}@<-0.4mm>[ul] \ar@^{-}@<0.0mm>[ul] \ar@^{->}@<0.4mm>[ul] 
}\ea.
\end{align}
\tc{Notice that, in this example, the cycle currents' direction is inverted in the black box and preserved along the observable edge, but their intensities cannot be exactly preserved as this would violate Kirchhoff's law of current conservation at the sites.} As regards Kirchhoff's loop law, we can show ({\bf Theorem 5}) that the HTR has affinities $- F_\alpha$ along all cycles that do not include $1-2$ and affinity $2\Q - F_\alpha$ for all cycles that do include $1-2$. Therefore, at a stalling steady state where $\Q = 0$ one has exact inversion of all the affinities and currents, and $\widetilde{\W}(x^\mathrm{st}) = \overline{\W}(x^\mathrm{st})$, which we call the condition of {\it marginal detailed balance}.

The proof of {\bf Theorem 3} relies on a direct  comparison of the path probabilities \cite{weber} $P(\omega)$ and $\widetilde{P}(\omega)$  of the original and the HTR dynamics, and by taking appropriate marginals. Another, more elegant way to derive the same result is via the Generating Function of the (time-)Scaled current's Cumulants (SCGF) $\lambda(q) = \langle \phi \rangle q + \frac{1}{2} \mathrm{var}(\phi) q^2 + \ldots$ in terms of which, after a bilateral Laplace transform, the marginal FR Eq.\,(\ref{eq:MFR}) takes the form of a {\it marginal Lebowitz-Spohn} symmetry \cite{ls}
\begin{align}
\widetilde{\lambda}(q) = \lambda(\Q - q). \label{eq:LS}
\end{align}
Let us outline the proof of this latter fact. To discuss possible generalizations, we allow for a current supported on several edges, $\phi = \sum_{ij} \varphi_{ij}\phi_{ij}$, with $\varphi_{ij} = - \varphi_{ji}$. It is well known that the SCGF is the Perron-Froebenius eigenvalue of the operator $\M(q)$ obtained by replacing off-diagonal entries of $\W$  by $w_{ij} e^{-\varphi_{ij} q}$, while keeping the diagonal exit rates. Then, letting $\R(q)$ be the diagonal matrix with entries the right eigenvector of $\M(q)$, one can construct an {\it auxiliary} Markovian generator $\W(q)= \R(q) \M(q)^T \R(q)^{-1} - \lambda(q) \mathbb{I}$ with the property that a given rare fluctuation of the current according to dynamics $\W = \W(0)$ is the typical current according to $\W(q)$, for some value of $q$ \cite{chetrite}. We show in {\bf Theorem 4} and corollaries that, for a single edge current, at $q =\Q$ the SCGF vanishes, $\lambda(\Q)= 0$, the HTR coincides with the auxiliary dynamics $\widetilde{\W} = \W(\Q)$, and the right eigenvector can be interpreted as stalling steady state $\R(\Q) = \P_\st$. Equation (\ref{eq:LS}) follows from matrix similarity $\Box$.

This formalism allows us to significantly improve on our results. We can go back to the issue of the long-time limit and relax the assumption showing that if the observer prepares the system in the stalling state $\vec{p}^{\;\st}$ by tuning $x$ to $x^{\st}$, and then suddenly turns it back on, the FR holds at all times (see {\bf Theorem 4}).

\begin{figure}
  \centering
  \includegraphics[width=0.95\columnwidth]{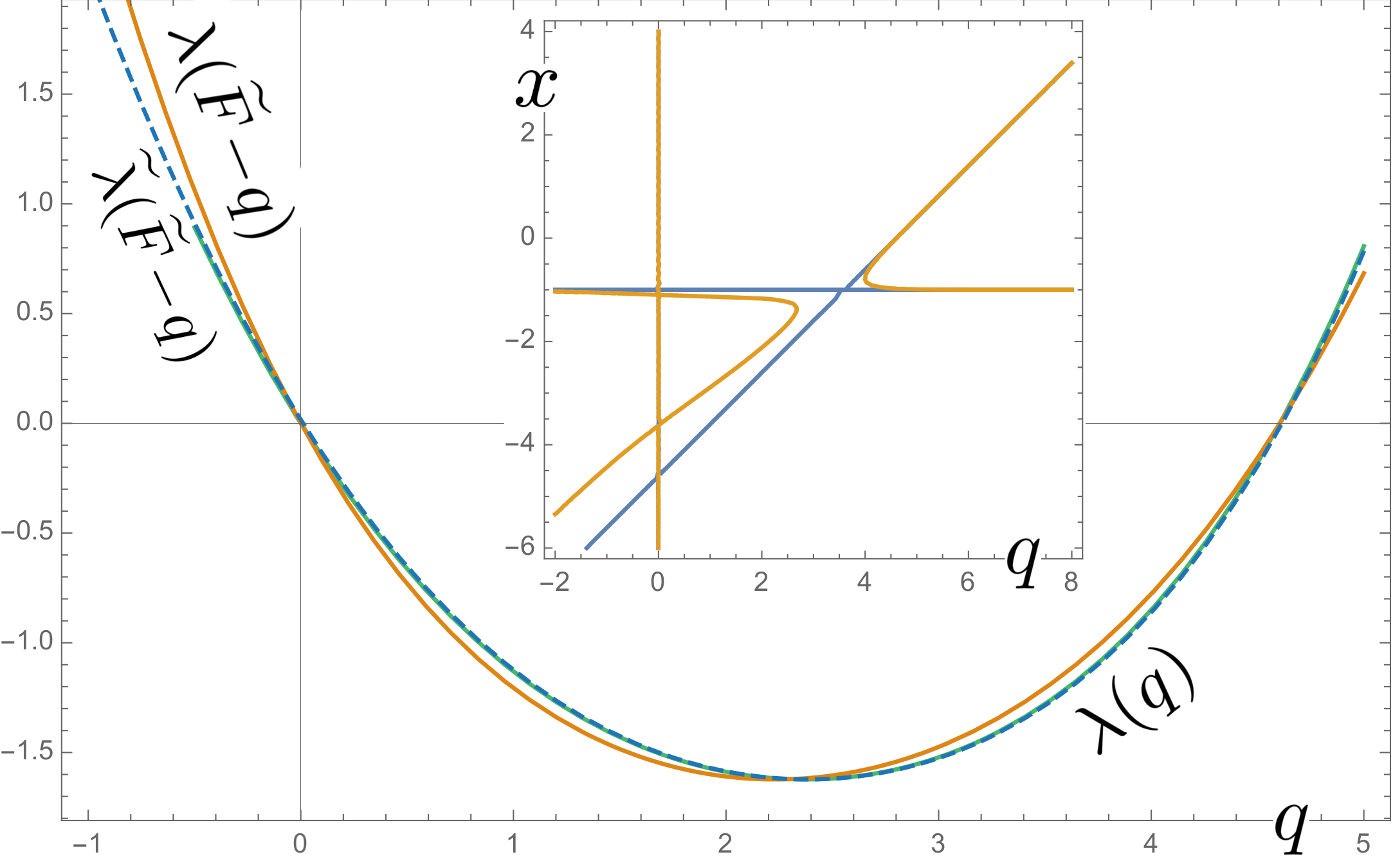} 
\caption{
We consider the model in Eq.\,(\ref{eq:model}) with rates
$1 = w_{31} = w_{14} = w_{32} = w_{43} = w_{21}/(1+x)$
and $10 = w_{13} = w_{41} = w_{23} = w_{34} e^{-y} = w_{12} e^{-x}/(1+x)$. {\bf Main frame}: For $x=y=0$, we plot the SCGF $\lambda(q)$ and its symmetric $\lambda(\Q-q)$; the two do not coincide and only meet at $0$, $\Q/2$ and $\Q$, showing that the FR does not hold. We also plot (dashed) $\widetilde{\lambda}(\Q-q)$, which perfectly coincides with $\lambda(q)$, showing that the FR is restored by the hidden-TR dynamics. {\bf Inset:} When only transition $1-2$ is counted, for $y = 0$, straight lines are the {\it loci} where $\lambda(q,x)\equiv 0$; apart from the trivial $q=0$ and $x=-1$ lines, the diagonal line corresponds to the effective affinity $\Q(x)$, showing that it is linear in $x$. The curvy lines are obtained when we also count $34$ [by setting  $\varphi_{ij} = \delta_{i,1}\delta_{j,2} + \delta_{i,3}\delta_{j,4} - (i \leftrightarrow j)$ and $y = x$]. The effective affinity is no longer linear in $x$.}
\label{fig:image}
\end{figure}

In Fig.\,\ref{fig:image} we plot the SCGF for the original and the HTR dynamics, and for their symmetric obtained by $q \to \Q -q $. The plot clearly shows that the marginal FR holds. The latest argument we outlined might suggest that our results could generalize to the case where several transitions contribute to the observable current; {\it mathematically}, all is needed is that there exists another value $q=\Q$ at which $\lambda(\Q) = 0$. However, in our previous {\it physical} construction it was crucial that the effective affinity was linear in the parameter $x$, and that $\R(\Q)$ could be interpreted as the stalling steady state. This is not the case in general, as shown in the inset in Fig.\,\ref{fig:image}. Hence one must be prudent in claiming generality.

The case of currents supported on several edges, where symmetries enter the picture \cite{polettini2016}, is cogent: a complete treatment is in preparation. However, we can turn the above considerations upside down: if marginal fluctuation and response relations do not work, then there is a nontrivial hidden interaction between the observable and hidden degrees of freedom.

\paragraph*{Aknowledgments.} 

This project could not possibly come to light without thorough discussions trough the years with Bernhard Altaner. The research was supported by the National Research Fund Luxembourg (project FNR/A11/02) and by the European Research Council, project NanoThermo (ERC-2015-CoG Agreement No. 681456).

\setcounter{equation}{0}

\renewcommand{\theequation}{S\arabic{equation}}

\newpage

\section*{Supplemental Material}

\begin{theorem}
\tc{The unique value of the} effective rates $\widetilde{w}_{12}, \widetilde{w}_{21}$ satisfying Eqs.\,(4) and (5) in the main text, for all values of $x$, is given by Eq.\,(7) in the main text.
\end{theorem}

\begin{proof}
\tc{Let us expand} the vanishing determinant of the generator with Laplace's formula \tc{along the $i$-th row}
\begin{align} 
0 = \det \W = \sum_{j} \W_{ij} (-1)^{i+j} \det \W_{(i\tc{|}j)}.
\end{align}
\tc{Comparing to the steady-state equation $\sum_j \W_{ij} p_j = 0$,} one finds that the steady-state probability can be written in terms of the minors of the generator as
\begin{align} 
p_j \propto (-1)^{i+j} \det \W_{(i\tc{|}j)}, 
\end{align}
independently of $i$. \tc{Let us consider the minors $\det \W_{(2\tc{|}1)}$ and $\det \W_{(1\tc{|}2)}$. Expanding the first with Laplace's formula along the second column, first row, we obtain}
\begin{align} 
\det \W_{(2\tc{|}1)}(x) = w_{12}(x) \det \W_{(1,2|1,2)} + \det \W^{\mathrm{st}}_{(2\tc{|}1)},
\end{align}
where $\W^{\mathrm{st}}$ is the {\it stalling} generator obtained from $\W$ by setting $w_{12} \equiv w_{21} \equiv 0$. \tc{We can proceed in a similar way  for $\det \W_{(1\tc{|}2)}$. A comment is in order:} By the all-minors matrix-tree theorem in algebraic graph theory \cite{chaiken}, minors can be given a graphical representation in terms of oriented rooted spanning trees, \tc{and the above formula can be interpreted as a {\it deletion-contraction formula}  \cite{sokal} stating that spanning trees either pass by edge $1-2$ (contraction) or do not (deletion)}.

We then obtain
\begin{align} 
& \frac{p_1(x)}{p_2(x)}
= \frac{\det \W_{(2\tc{|}1)}(x)}{\det \W_{(1\tc{|}2)}(x)} \label{eq:dede} \\ 
& = \frac{ \Scale[0.5]{
  \ba{c}\xymatrix{ \ar@{<-}[r] &  \ar@{<-}[d] \\  \ar[u] & } \ea
+ \ba{c}\xymatrix{ \ar@{<-}[r] & \ar@{<-}[d] \\ & \ar@{<-}[l]} \ea
+ \ba{c}\xymatrix{ \ar@{<-}[r] &  \\ \ar[u] & \ar[l]} \ea
+ \ba{c}\xymatrix{ \ar@{<-}[r]  & \\  \ar[r] & \ar[ul]  } \ea
+ \ba{c}\xymatrix{    &  \ar[l] \\ \ar[u] & \ar@{->}[ul]  } \ea 
+ \ba{c}\xymatrix{ & \ar[d] \\ \ar[u] & \ar[l]} \ea 
+ \ba{c}\xymatrix{ & \ar@{->}[d] \\ \ar[u] & \ar@{->}[ul] } \ea
+ \ba{c}\xymatrix{   & \ar[d] \\  \ar[r] & \ar[ul] } \ea 
} }{\Scale[0.5]{
\ba{c}\xymatrix{ \ar@{->}[r] &  \ar@{<-}[d] \\  \ar[u] & } \ea
+ \ba{c}\xymatrix{ \ar@{->}[r] & \ar@{<-}[d] \\ & \ar@{<-}[l]} \ea
+ \ba{c}\xymatrix{ \ar@{->}[r] &  \\ \ar[u] & \ar[l]} \ea 
+ \ba{c}\xymatrix{ \ar@{->}[r]  & \\  \ar[r] & \ar[ul]  } \ea
+ \ba{c}\xymatrix{   \ar[r]  &  \\ \ar[u] & \ar@{->}[ul]  } \ea
+ \ba{c}\xymatrix{ & \ar@{<-}[d] \\ \ar@{<-}[u] & \ar@{<-}[l]} \ea
+ \ba{c}\xymatrix{ & \ar@{<-}[d] \\ \ar[u] & \ar@{<-}[ul] } \ea
+ \ba{c}\xymatrix{   & \ar@{<-}[d] \\  \ar[r] & \ar@{<-}[ul]  } \ea
}} \nonumber
\end{align}
where we also provided the graphical representation. Notice that there are five trees that pass through edge $1-2$ and three trees that do not.

Comparing Eq.\,(\ref{eq:dede}) with Eq.\,(5) in the main text we obtain
\begin{align}
\frac{w_{21}(x) + \widetilde{w}_{21}}{w_{12}(x) + \widetilde{w}_{12}} = \frac{ w_{12}(x) \det \W_{(1,2|1,2)} + \det \W^{\mathrm{st}}_{(2\tc{|}1)}}{ w_{21}(x) \det \W_{(1,2|1,2)} + \det \W^{\mathrm{st}}_{(1\tc{|}2)}} . \label{eq:condition}
\end{align}
Since this must hold for all $x$, we conclude.
\end{proof}

\begin{theorem}
The operator $\widetilde{\W}$ is the generator of a Markov jump process.
\end{theorem}

\tc{
\begin{proof} The general form of a Markov jump-process generator on a finite state space is a square matrix with positive off-diagonal entries and columns adding up to zero. In our case, diagonal entries are negative, and off-diagonal entries are indeed positive, since all entries of $\W_{\mathrm{st}}$ are at most as large as the entries of $\W$. We then only need to check that columns add up to zero. We obtain
\begin{align} 
\sum_i \widetilde{\W}_{i1} & = -\sum_i w_{i,1} + w_{21} + \frac{1}{p_1^{\st}} \sum_{i>2} w_{1i} \, p_i^{\st}, \\  
\sum_i \widetilde{\W}_{i2} & = -\sum_i w_{i,2} + w_{12} + \frac{1}{p_2^{\st}} \sum_{i>2} w_{2i} \, p_i^{\st}, \\
\sum_i \widetilde{\W}_{ij} & =  \frac{1}{p_j^{\st}} \sum_{i} \W^{\st}_{ji} \, p_i^{\st},  \qquad j > 2 
\end{align}
In fact, a close inspection reveals that the third equation also includes the first two, once we recognize the definition of the stalling generator $\W^{\st}$. The expression vanishes because $\vec{p}^{\;\st}$ is its null eigenvector.
\end{proof}
}

\begin{theorem}
The marginal FR Eq.\,(12) in the main text holds.
\end{theorem}

\begin{proof}
Consider a trajectory, viz. a single realization of a continuous-time Markov jump process
\begin{align} 
\omega = (i_1,t_1) \to (i_2,t_2) \to \ldots \to (i_N, t_N),
\end{align}
depicted by a successions of sites and of waiting times before a transition to a new site occurs, up to time $t = \sum_{n=1}^N t_n$. To momentarily dispense with boundary terms, irrelevant at long-enough time, we assume that the trajectory is cyclic $i_N \equiv i_0$. Along trajectories, one can define observables. In particular, we are interested in currents, which are the random variables counting the net number of transitions $ij$ from $j$ to $i$:
\begin{align} 
\phi^{\omega}_{ij} = \sum_{n=0}^{N-1} \left(\delta_{i,i_{n+1}} \delta_{j,i_n} -  \delta_{j,i_{n+1}} \delta_{i,i_n}\right).
\end{align}

To any generator of a continuous-time, discrete-state-space Markov jump process is associated a well defined p.d.f. $P(\omega)$ over trajectories $\omega$. Given the initial site, the p.d.f. is a sequence of exponentially distributed waiting times and of instantaneous jump rates (see \cite{weber} for a review)
\begin{align} 
P(\omega) =   e^{-w_{i_N} t_N} \prod_{n=0}^{N-1} w_{i_{n+1},i_n} \, e^{-w_{i_n} t_n} .\label{eq:densityhtr}
\end{align}
Similarly, we can compute $\widetilde{P}(\omega)$. Since the forward and the HTR generators have the same exit rates, the waiting-time terms are exactly the same. Taking the ratio of the path probabilities, we obtain
\begin{align} 
\frac{\prob(\omega)}{\widetilde{\prob}(\omega)} = \prod_{n=0}^{N-1} \frac{w_{i_{n+1},i_n}}{\widetilde{w}_{i_{n+1},i_n}} = \prod_{ij\neq 12, 21} \left( \frac{w_{ij} p^\st_j}{w_{ji}p^\st_i}\right)^{\phi^\omega_{ij}/2} \label{eq:tprob}
\end{align}
\tc{where $\sum_{ij}$ denotes the sum over oriented edges, to be distinguished from the sum $\sum_{i,j}$ over couples of sites (the latter are related to the former by the handshaking lemma in graph theory $\sum_{i,j} = 2 \sum_{ij}$)}. Notice that in Eq.\,(\ref{eq:tprob}) we used the fact that hidden TR and original dynamics have the same rates along edge $1-2$, which then cancel out in the ratio.

We can now multiply by
\begin{align} 
1 =  e^{-\Q\, \phi^{\omega}} \big(\nicefrac{w_{12}p_2^\st}{w_{21} p_1^\st}\big)^{\phi^{\omega}}
\end{align}
to obtain
 \begin{align} 
\frac{\prob(\omega)}{\widetilde{\prob}(\omega)} & =  e^{-\Q\, \phi^{\omega}} \prod_{ij} \left( \frac{w_{ij} }{w_{ji}}\right)^{\phi^\omega_{ij}} =  e^{-\Q\, \phi^{\omega}}  \frac{P(\omega)}{\overline{P}(\omega)}.
\end{align}
In the first passage, the contributions $p_{i_n}^\st/p_{i_{n+1}}^\st$ cancel out because the trajectory is continuous and cyclic. In the second passage we used the complete FR \cite{polettiniFT}
\begin{align} 
\frac{P(\omega)}{\overline{P}(\omega)} = \exp \sum_{ij}  \phi^\omega_{ij}  \log \frac{w_{ij} }{w_{ji}},
\end{align}
that relates the time forward and time reversed probabilities. \tc{Rearranging terms we obtain
\begin{align} 
\overline{P}(\omega)=  e^{-\Q\, \phi^{\omega}}  \widetilde{\prob}(\omega).
\end{align}
We can now integrate over all trajectories that yield current $\phi^{\omega}= -\phi$, to obtain the FR 
$P(\phi)=  e^{\Q\, \phi}  \widetilde{\prob}(-\phi)$, where we used the fact that $\overline{P}(-J) = P(J)$.  
}
\end{proof}

\begin{theorem}
The HTR can be expressed in terms of the tilted operator $\M(q)$ with counting field $q$ associated to the current through edge $1-2$ as
\begin{align} 
\widetilde{\W} = \P_\st \, \M(\Q)^T \, \P_\st^{-1} \label{eq:33}
\end{align}
As consequences, the marginal Lebowitz-Spohn symmetry Eq.\,(17) holds, and the FR Eq.\.(12) holds at all times, provided the initial distribution is sampled from the stalling state.
\end{theorem}

\begin{proof}
The tilted operator $\M(q)$ is obtained by replacing off-diagonal entries $w_{12}$, $w_{21}$ of the generator by $w_{12} e^{-q}$, $w_{12} e^{+q}$ while keeping the diagonal exit rates, e.g.
\begin{align} 
\M(q) = \left(\ba{cccc} - w_1 & w_{12} e^{-q} & w_{13} & w_{14} \\  w_{21} e^{q} & -w_2  & w_{32} & 0 \\  w_{13}   & w_{32}    & - w_3 & w_{34}   \\  w_{41}  & 0 & w_{43}   & -w_4 \ea \right).
\end{align}
Then, identity Eq.\,(\ref{eq:33}) can be proven by direct calculation and by comparison with Eq.\,(15).

Let us now compute the SCGF of the HTR process. Again, it is the Perron-
Froebenius eigenvalue of the operator $\widetilde{\M}(q)$ obtained from $\widetilde{\W}$ by replacing $w_{12}$, $w_{21}$ with $w_{12} e^{-q}$, $w_{12} e^{+q}$. We notice that
\begin{align} 
\widetilde{\M}(q) = \P_\st \M(\Q -q)^T \P_\st^{-1}. \label{eq:fintim}
\end{align}
Then, by matrix similarity one concludes that the spectra of the two matrices coincide, and therefore the marginal Lebowitz-Spohn symmetry holds.

As a remark, it is well known \cite{chetrite} that for all values of $q$ one can build a new Markov generator $\W(q)$ out of $\M(q)$ by Doob's transform $\widetilde{\W} = \R(q) \M(q)^T \R^{-1}(q) - \lambda(q) \mathbb{I}$, where $\R(q)$ is the diagonal matrix containing the entries of the right-null  eigenvector of $\M(q)$. Hence this implies that $\widetilde{\W}$ is the Doob transform evaluated at $\lambda(\Q) = 0$, and furthermore that $\R(\Q) = \P_\st$.

Finally, let us prove that the above relation Eq.\,(\ref{eq:fintim}) implies that the FR holds at all times, along the lines of the proof of the FR in the Appendix of Ref.\,\cite{polettiniFT}. Consider the moment generating function of the current $\zeta(q,t)$ at time $t$, and let $\zeta_i(q,t)$ be the moment generating function conditioned to being at site $i$ at time $t$, such that $\zeta(q,t) = \vec{1}\cdot \vec{\zeta}(q,t)$, where $\vec{1}$ is the vector with all $1$'s \tc{as entries}. It is well known that the latter conditional moment generating function, and the corresponding one for the HTR dynamics, evolve by
\begin{subequations}
\begin{align} 
\frac{d}{dt} \vec{\zeta}(q,t) & = \M(q)\vec{\zeta}(q,t), \\
\frac{d}{dt} \vec{\widetilde{\zeta}}(q,t) & = \widetilde{\M}(q)\vec{\widetilde{\zeta}}(q,t).
\end{align}
\end{subequations}
Let $\vec{\zeta}(q,0) = \vec{\widetilde{\zeta}}(q,t) \equiv \vec{p}(0)$ be the initial configuration of the system. Defining $\U(q,t) = \exp t \M(q)$ and $\widetilde{\U}(t) = \exp t \widetilde{\M}(q)$ we have
\begin{align} 
\vec{\zeta}(q,t) & =  \U(q,t)\vec{p}(0).
\end{align}
Similarly, for the HTR dynamics, evaluating at $q\to \Q-q$ we obtain
\begin{align} 
\vec{\widetilde{\zeta}}(\Q-q,t) & = \widetilde{\U}(\Q - q,t)\vec{p}(0) \\
& = \P_\st \left[\exp t\M(q)^T \right] \P_\st^{-1}  \vec{p}(0)  
\end{align}
where in the latter passage we used Eq.\,(\ref{eq:fintim}). We now evaluate the moment generating functions as
\begin{subequations}
\begin{align} 
\zeta(q,t) & = \vec{1}\cdot \U(q,t)\vec{p}(0) \\
\widetilde{\zeta}(q,t) & = \vec{p}(0) \cdot \P_\st^{-1} \U(q,t) \P_\st \vec{1} 
\end{align} 
\end{subequations}
from which it follows that the choice of $\vec{p}(0) = \vec{p}^{\,\st}$ \tc{makes the moment generating functions} of the current for the forward and for the HTR dynamics identical at all times.
\end{proof}

\begin{theorem}
Let $F(\C)$ and $\widetilde{F}(\C)$ denote the cycle affinities with respect to the forward and marginal dynamics, respectively. We have
\begin{subequations}
\begin{align} 
\widetilde{F}(\C) &= - F(\C) , & & \mathrm{if}\,  \C \niton 1-2  \\
\widetilde{F}(\C) &= - F(\C) + 2\Q , & & \mathrm{if}\, \C \owns 1-2. 
\end{align}
\end{subequations}
\end{theorem}

\begin{proof}
For all cycles that do not contain edge $1-2$ one has
\begin{align}
\widetilde{F}(\C) & = \log \prod_{ij \in \C}  \frac{\widetilde{w}_{ij}}{\widetilde{w}_{ji}} = \log \prod_{ij \in \C}  \frac{w_{ji}  p^\st_i/p^\st_j }{w_{ij} p^\st_j/p^\st_i}  
\end{align}
All $p^\st_i$ terms \tc{cyclically} cancel out one \tc{another}, so that the last term can be identified with $-F(\C)$. As regards cycles containing edge $1-2$, one has
\begin{align} 
\widetilde{F}(\C) & =\log \prod_{\substack{ij \in \C \\ ij \neq 12}}    \frac{w_{ji}  p^\st_i/p^\st_j }{w_{ij} p^\st_j/p^\st_i}  + \log \frac{w_{12}}{w_{21}} \\ 
& = \log \prod_{ij \in \C}  \frac{w_{ji}  (p^\st_i)^2}{w_{ij} (p^\st_j)^2} - \log \frac{w_{21}  (p^\st_1)^2}{w_{12} (p^\st_2)^2} + \log \frac{w_{12}}{w_{21}}
\end{align}
and we conclude.
\end{proof}

\newpage

$\;$

\end{document}